\documentclass[11pt]{article}%

\usepackage{cleveref, bm, commath}
\usepackage{graphicx,tabularx}
\usepackage{subcaption}
\usepackage{geometry}

\usepackage{natbib}

\usepackage{multirow}
\usepackage{algorithmic}

\usepackage{color}
\definecolor{wjs}{RGB}{0,0,255}

\definecolor{qz}{RGB}{0,255,0}

\usepackage[plain,noend]{algorithm2e}
\usepackage{amsmath, amsfonts, amsthm}

\usepackage{xr}

\usepackage{sectsty}
\usepackage[onehalfspacing]{setspace}

\newtheorem{theorem}{Theorem}
\newtheorem{proposition}{Proposition}
\newtheorem{lemma}{Lemma}

\crefname{assumption}{Assumption}{Assumptions}
\crefname{equation}{equation}{equations}

\theoremstyle{definition}
\newtheorem{definition}{Definition}

\newcommand\independent{\protect\mathpalette{\protect\independenT}{\perp}}
\def\independenT#1#2{\mathrel{\rlap{$#1#2$}\mkern2mu{#1#2}}}

\newcommand*\diff{\mathop{}\!\mathrm{d}}

\makeatletter
\renewcommand{\algocf@captiontext}[2]{#1\algocf@typo. \AlCapFnt{}#2} 
\def\@algocf@capt@plain{top}
\renewcommand{\algocf@makecaption}[2]{%
  \addtolength{\hsize}{\algomargin}%
  \sbox\@tempboxa{\algocf@captiontext{#1}{#2}}%
  \ifdim\wd\@tempboxa >\hsize
  \hskip .5\algomargin%
  \parbox[t]{\hsize}{\algocf@captiontext{#1}{#2}}
  \else%
  \global\@minipagefalse%
  \hbox to\hsize{\box\@tempboxa}
  \fi%
  \addtolength{\hsize}{-\algomargin}%
}
\makeatother

\begin{document}
%
\sectionfont{\bfseries\large\sffamily}%
%

\subsectionfont{\bfseries\sffamily\normalsize}%
%

\noindent
{\sffamily\bfseries\Large
Multiple testing when many $p$-values are uniformly
    conservative, with application to testing qualitative interaction in
    educational interventions}%
%

\noindent
\textsf{Qingyuan Zhao\footnote{\textit{Address for correspondence:} Department of Statistics, The
Wharton School, University of Pennsylvania, Jon M. Huntsman Hall, 3730 Walnut
Street, Philadelphia, PA 19104-6340 USA. \ \textsf{E-mail:}
qyzhao@wharton.upenn.edu. \ 23 August 2017.}, Dylan S. Small, Weijie Su}%
%

\noindent
\textsf{University of Pennsylvania, Philadelphia}%

\noindent

\textsf{Abstract. \ In the evaluation of treatment effects, it is of
  major policy interest to know if the
  treatment is beneficial for some and
    harmful for others, a phenomenon known as qualitative
    interaction. We formulate this question as a multiple testing
    problem with many conservative null $p$-values, in which the
    classical multiple testing methods may lose power substantially. 
    We propose a simple technique---conditioning---to improve the
    power. A crucial assumption we need is uniform conservativeness,
    meaning for any conservative $p$-value $p$, the conditional distribution $(p/\tau)\,|\,p \le
    \tau$ is stochastically larger than the uniform distribution on
    $(0,1)$ for any $\tau$. We show this property holds for one-sided tests
    in a one-dimensional exponential family (e.g.\ testing for qualitative
    interaction) as well as testing
    $|\mu|\le\eta$ using a statistic $X \sim \mathrm{N}(\mu,1)$ (e.g.\
    testing for practical importance with threshold $\eta$). We
    propose an adaptive method to select the threshold $\tau$. Our
    theoretical and simulation results suggest the proposed tests
    gain significant power when many $p$-values are uniformly conservative and lose
    little power when no $p$-value is uniformly conservative. We apply
    our method to two educational intervention datasets.}%


\newpage








\section{Introduction}
\label{sec:introduction}

To make the most informed policy decisions from a randomized
experiment or an observational study, it is often important to
understand variation in treatment effects
\citep{wang2007statistics,schochet2014understanding}. For example, an
influential framework in education called Aptitude-Treatment
Interaction \citep{cronbach1977aptitudes} 
is based on the claim that different learners may benefit from
different styles of instruction. There are also
numerous examples in medical studies. \citet{pizzocaro2001interferon}
found that a drug class called interferon increases the survival probability
of patients in the late stage of renal cell carcinoma, but is harmful
for patients in the early stage. In the evaluation of a randomized
trial \citep{crash2011importance}, the collaboration team found that
tranexacmic acid, the drug under study, reduces the risk of bleeding to death when
used within 3 hours from trauma injury, but increases the risk when
used after 3 hours.

Typically, the heterogeneous nature of treatment effect is examined by
subgroup analysis, where study participants are
grouped by some baseline covariates measured prior to the
treatment. We will use the terms ``treatment effect heterogeneity'' and
``treatment interaction'' interchangeably because both of them means there is a
noticeable interaction between the treatment and the subgroup
indicator when modeling the outcome. There are two types of treatment effect
heterogeneity: the \emph{qualitative} or \emph{disordinal}
interaction, where there exists one subgroup whom
the treatment benefits and another subgroup whom the treatment
harms, and the \emph{quantitative} or \emph{ordinal} interaction,
where the subgroup treatment effects may have different magnitude but
the same sign \citep{gail1985testing}. A good example
of qualitative interaction is the study of interferon described in the
last paragraph. There is also interest in understanding
treatment effects variation across sites in multisite studies or across
studies in meta-analysis \citep{bloom2017using}.

\subsection{Motivating applications}
\label{sec:motivating-examples}

In this paper we will consider the problem of testing qualitative
interaction. Since qualitative
interactions imply the optimal treatment rule must be personalized,
they usually have much more policy or clinical significance than
non-qualitative interactions. To motivate this, we first introduce two examples from
education. In the first example, we consider the effect of modified school
calendars. Instead of
following the more traditional school calendar with a long summer
break (in addition to a short winter and spring break), some schools
have switched to a modified school calendar comprising more frequent
but shorter intermittent breaks (e.g., 9 weeks of school followed by
3 weeks off), while keeping the total number of days at school
approximately the same. \citet{cooper2003effects} investigated the
effect of modified school calendars on student achievement using
studies of $55$ schools in $11$ districts. Using the published dataset
in \citet{konstantopoulos2011fixed}, we summarize their main results
in \Cref{fig:meta-analysis-1}.

In the second example, we consider the effectiveness of
writing-to-learn interventions, in which students receive
instruction with increased emphasis on writing tasks compared to
conventional instruction. The outcome of interest is the academic
achievement of the students (for example some exam score).
\Cref{fig:meta-analysis-2} summarizes the results of a meta-analysis
of $48$ studies by \citet{bangert2004effects}.

In both examples, the treatment has an
overall significantly positive effect when fitting a random effect
model (the ``RE Model'' row in \Cref{fig:meta-analysis}). However, for
the modified school calendar,  a multi-level analysis in
\citet{konstantopoulos2011fixed} showed significant heterogeneity in
the treatment effect; see also
\Cref{sec:meta-analysis-1}. Furthermore, in both examples there
is at least one significantly negative finding (without correcting
for multiple testing) in the forest plot. Therefore, it is
interesting to know if qualitative interaction exists, that is, if for
some cohorts the treatment effect is indeed negative.

\begin{figure}[t]
  \centering
  \begin{subfigure}[b]{0.49\textwidth}
    \includegraphics[width =
    \textwidth]{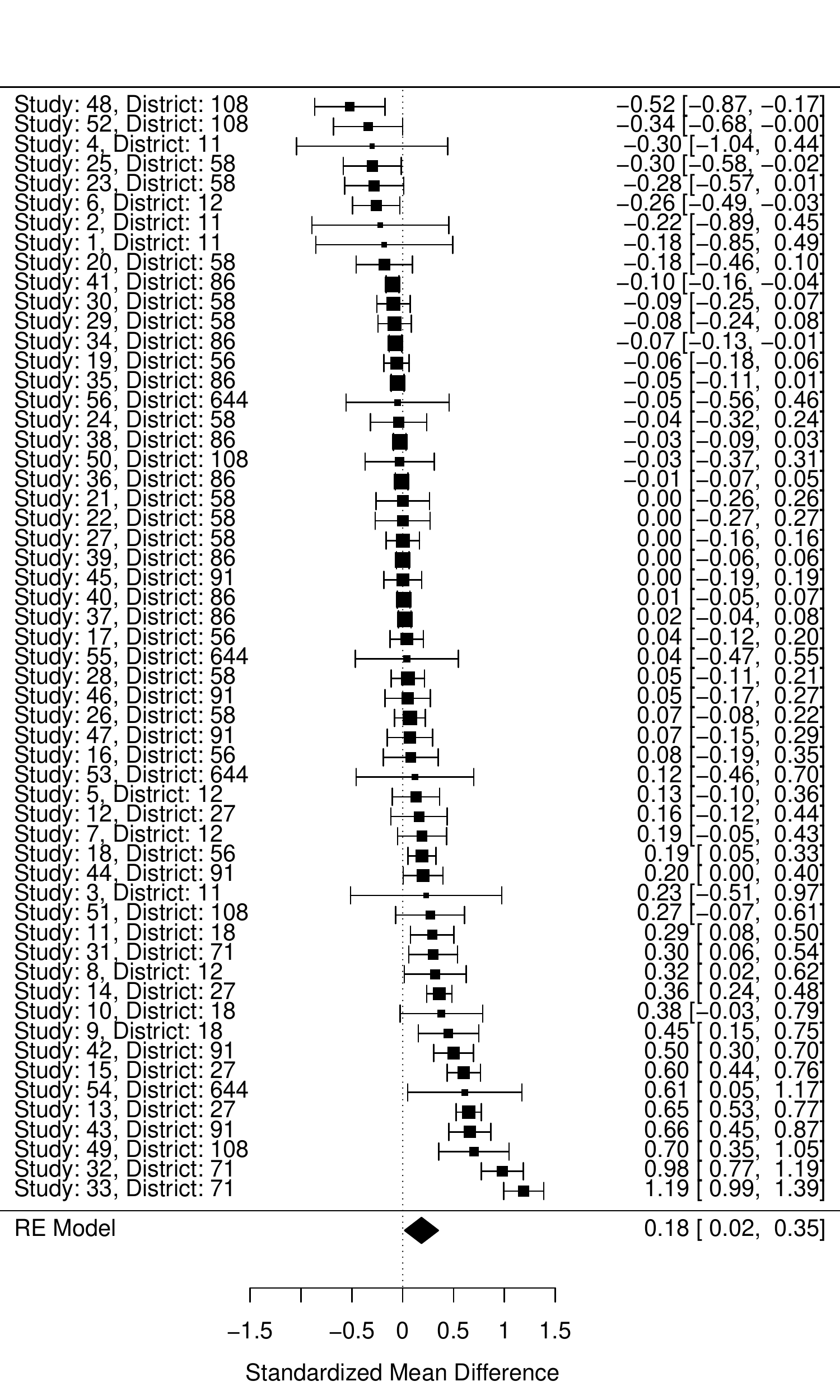}
    \caption{Example 1: effect of modified school calendar.}
    \label{fig:meta-analysis-1}
  \end{subfigure}
  \begin{subfigure}[b]{0.49\textwidth}
    \includegraphics[width =
    \textwidth]{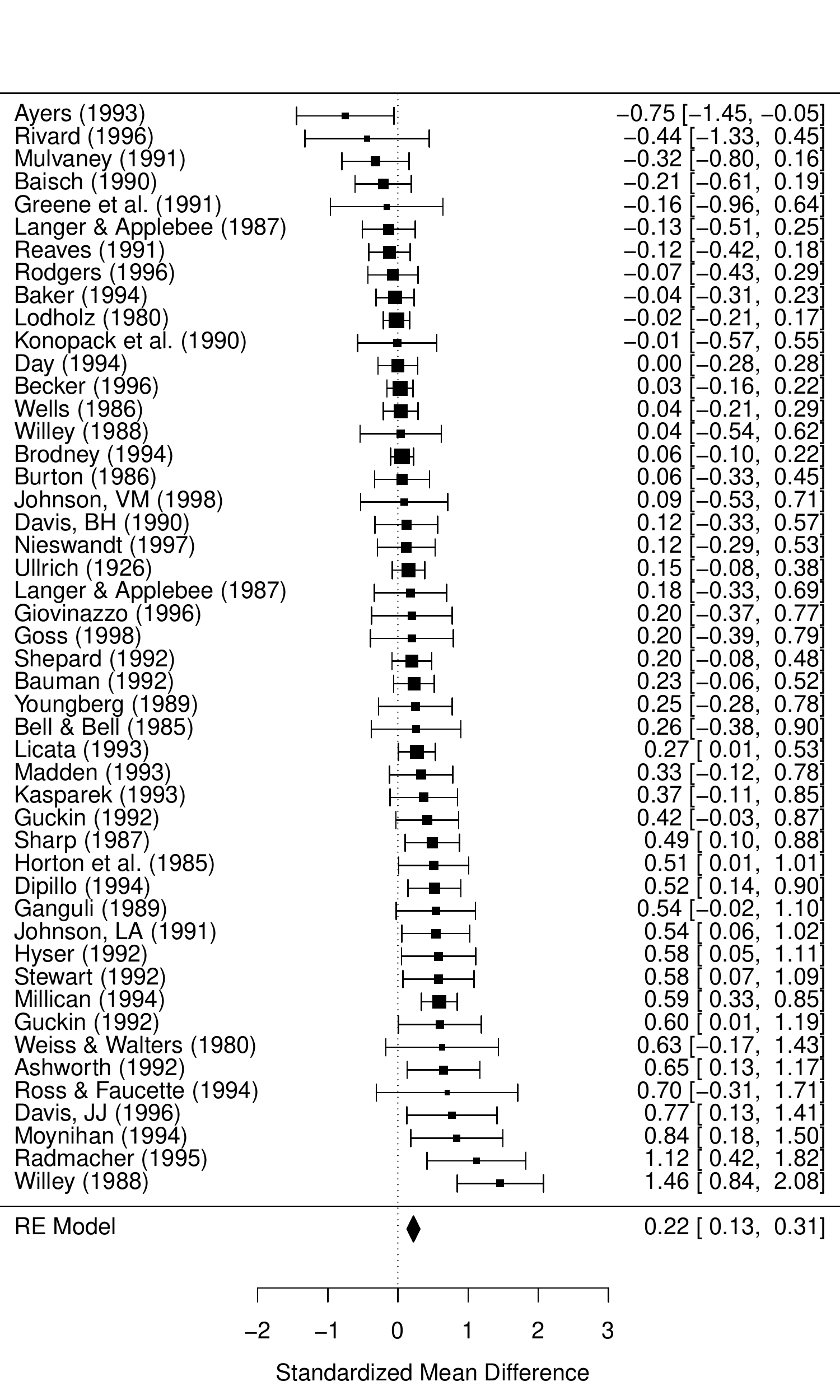}
    \caption{Example 2: effect of writing-to-learn intervention.}
    \label{fig:meta-analysis-2}
  \end{subfigure}
  \caption{Forest plots of two meta-analyses with potential
    qualitative interaction.}
  \label{fig:meta-analysis}
\end{figure}

\subsection{Formulation of testing qualitative interaction as a
  multiple testing problem}
\label{sec:form-test-qual}

Next we formulate this question as a multiple testing
problem. Suppose we want to analyze a randomized experiment or an
observational study in which there are $n$ subgroups or independent
studies and the observed treatment effect in study $i$ is distributed
as $\mathrm{N}(\mu_i,\sigma_i^2)$. Then the null hypothesis of no qualitative
interaction is the union of two hypotheses, $H_0 = H_0^{+} \cup
H_0^{-}$, where $H_0^{+} = \{\mu_j \ge 0,
\forall i\}$ and $H_0^{-} = \{\mu_i \le 0, \forall i\}$. The one-sided
hypothesis $H_0^{+}$ is the intersection of $n$ individual hypotheses,
$H_0^{+} = \cap_{i=1}^n H_{0i}^{+}$ where $H_{0i}^{+} = \{\mu_i \ge
0\}$, and is called a \emph{global
  null hypothesis} because $H_0^{+}$ is correct if and only if all the
individual hypotheses are correct. To test the null hypothesis $H_0$
at level $\alpha$, one can simply test both $H_0^{+}$ and $H_0^{-}$
and reject $H_0$ if both $H_0^{+}$ and $H_0^{-}$ are rejected at
level $\alpha$ (because $H_0$ is the \emph{union} hypothesis, we do not need
to correct for multiplicity here).


Our statistical methodology only requires $p$-values for the
individual hypotheses, $p_i$, $i=1,\dotsc,n$. In the classical
setting of a multiple testing problem, the $p$-values are
assumed to either follow the uniform distribution or be
stochastically larger than the uniform distribution over $(0,1)$ under
the null, or stochastically smaller than the uniform
distribution under the alternative. Many effective procedures have
been proposed and some have certain theoretical optimality
guarantees especially if the null $p$-values are exactly uniformly
distributed. Some distinguished examples of testing the global null
hypothesis include Bonferroni's correction, Fisher's
combination, and Tukey's higher criticism, which are reviewed in
\Cref{sec:previous-methods}. When the global null is rejected, usually
we are further interested in knowing which individual hypotheses
(namely $H_{0i}^{+}$ and $H_{0i}^{-}$) are false. Subsequent methods such as
\citet{holm1979simple}'s procedure or
\citet{benjamini1995controlling}'s procedure can be used to control
the family-wise error rate or the false discovery rate.

However, in the motivating examples above as well as many other applications,
it is common that the majority of the null $p$-values may be very
conservative (stochastically larger than the uniform
distribution). For example, for testing $H_{0i}^{+}$, if the
individual effect $\mu_i$ is positive, then the corresponding
$p$-value will be conservative. In both motivating examples, many of
the estimated individual effects are positive as shown in
\Cref{fig:meta-analysis}, and it seems plausible that many of the
effects are truly positive which would result in conservative
$p$-values for testing $H_0^{+}$. Ideally, if we
knew in priori that some studies have positive effects, we would like to
exclude them when testing $H_0^{+}$.

\subsection{Overview of our approach}
\label{sec:overview-paper}

In this paper we propose a simple
technique---conditioning---that can
be used in conjunction with all the existing multiple testing
procedures to improve power when many null $p$-values are
conservative. This avoids paying an unnecessary price of adjusting for
multiplicity for the conservative tests (see
\Cref{sec:probl-cons-tests} for some heuristics). The proposed method
can be applied to any global and simultaneous testing problem if the
following assumptions are satisfied: 1.\ the tests are
independent; 2.\ the $p$-values are \textit{uniformly valid}
(meaning the null $p$-values are still valid after conditioning).
When these two assumptions hold, we show, by
providing theoretical results and extensive simulation
studies, that the conditional tests reduce little power in the
classical non-conservative scenario and greatly increase power in the
conservative scenario (including testing qualitative
interaction). Detailed proofs of the claims in this paper can be found
in the appendix.

\subsection{Uniform validity and its applications}
\label{sec:unif-valid}

Among the two key assumptions, the first independence assumption can
be relaxed (\Cref{sec:valid-cond-test}), while we show the second
assumption (uniform validity) holds for one-sided tests in a family of
distributions that has monotone likelihood ratio (MLR). This includes one-sided tests in
any one-dimensional exponential family and hence includes the
motivating problem of testing qualitative interactions. The location family of
folded normal distribution also has MLR (see
\Cref{sec:fold-norm-distr}). This leads to another
application of the proposed conditional test---testing for practical
importance. In many problems, rejecting small
deviations from no effect is often practically inconsequential and
instead we would like to test whether there is a practically
meaningful difference from no effect. \citet{sun2012multiple}
formulated this problem as a two-sided normal means problem. Let $X_i
\sim \mathrm{N}(\mu_i,\sigma_i^2)$ with known $\sigma_i^2$. The $i$-th null
hypothesis is $H_{0i}: |\mu_i| \le \eta$ where $\eta$ is the practical
importance threshold. This can be viewed as a one-sided testing
problem in the folded normal distribution (if $X \sim
\mathrm{N}(\mu,\sigma^2)$, then $|X|$ is said to have a folded normal
distribution with parameters $\mu$ and
$\sigma^2$). Since the location family
of folded normal distributions (fixed $\sigma^2$) has MLR. Therefore,
our conditional test is also valid for the two-sided normal means
problem.

Although uniform validity hold for many hypothesis testing
problems as illustrated above, it does not hold in all
circumstances. We refer the reader to \Cref{sec:uniform-validity} for
an example and alternative approaches.




\section{Multiple testing in presence of conservative tests}
\label{sec:procedure}

Suppose we have $n$ $p$-values, $p_i$, $i=1,\dotsc,n$ for
$n$ hypotheses, $H_{0i},\,i=1,\dotsc,n$. We assume that
every null $p$-value $p_i$ is individually \emph{valid}, meaning
$\mathrm{P}(p_i \le q) \le q$ for all $0 \le q \le 1$ if $H_{0i}$ is
true. Furthermore, we call $p_i$
\emph{exact} if $\mathrm{P}(p_i \le q) = q$ for all $0 \le q \le 1$,
and \emph{conservative} if it is valid and $\mathrm{P}(p_i \le q) < q$
for some $0 < q < 1$. Geometrically if we plot the cumulative
distribution function (CDF) of a $p$-value, an exact $p$-value has a
CDF that is exactly the 45 degree diagonal line (the CDF of the uniform
distribution) and a conservative $p$-value'CDF is below the 45
degree line in at least one point.

\subsection{Previous methods for testing the global null}
\label{sec:previous-methods}

We will start with testing the global null hypothesis $H_0 = \bigcap_{i=1}^n
H_{0i}$ that all the individual hypotheses are true before moving into
other objectives (controlling family-wise error rate or false
discovery rate). Let's first review some classical methods to test the
global null. Given a significance level $0 < \alpha < 1$, one
of the simplest and most widely used methods is Bonferroni's
correction, which rejects $H_0$ if any of the $p_i$ is less than
$\alpha / n$. When testing the ``needle in a haystack problem'' (only
one non-null), this method is asymptotically optimal
\citep{arias2011global}. \citet{simes1986improved}
proposed an improved Bonferroni
procedure that rejects $H_0$ if $p_{(i)} \le
(i/n)\alpha$ for some $1 \le i \le n$, where $p_{(1)} \le p_{(2)} \le
\dotsb \le p_{(n)}$ are the ordered $p$-values. Notice that
Bonferroni's and Simes' procedures control type I error
for testing $H_0$ even if the $p$-values are not independent.

When the $p$-values are independent, another commonly used method is
Fisher's combination probability test,
which combines the $p$-values into one test statistic $T = - 2
\sum_{i=1}^n \log p_i$. \citet{fisher1925statistical} showed that $T$
has a $\chi^2$ distribution with $2n$ degrees of freedom when all the
$p$-values are uniformly distributed. When some hypotheses are false
the corresponding $p_i$ tend to be small, so the test statistic will
be large. A similar method is the truncated product of
\citet{zaykin2002truncated}, whose combined test statistic is
$T^{\prime} = - 2 \sum_{i=1}^n (\log p_i) 1_{\{p_i \le \tau\}}$,
where
$\tau$ is a truncation threshold between $0$ and
$1$. 

A third type of method compares the empirical distribution of
$p_1,\dotsc,p_n$ with the uniform distribution. An interesting
representative is Tukey's higher criticism or second-level
significance test, which examines if there is
an excessive number of significant tests (e.g.,\ tests with $p$-values
less than 0.05). \citet{donoho2004higher} considered a modified
statistic:
\[
  HC^{*} = \max_{0 < q \le \tau} \frac{n^{-1}\sum_{i=1}^n 1_{\{p_i \le q\}} - q}{\sqrt{q(1-q)/n}}.
\]
Here $n^{-1}\sum_{i=1}^n 1_{\{p_i \le q\}}$ and $q$ are the observed
and expected fractions of tests significant at level $q$, and
$\sqrt{q(1-q)/n}$ is the variance of the observed fraction under
$H_0$. \citet{donoho2004higher} showed that this test is very
effective at solving the ``sparse normal means'' problem in certain
asymptotic regime.

\subsection{Problem of conservative tests}
\label{sec:probl-cons-tests}

This paper considers the situation that the $p$-values are independent
but many of them are conservative. As an example, consider the
one-sided testing problem of normal means. Let $Y_i,\,i=1,\dotsc,n$ be
independent normal variables with mean $\mu_i$ and variance $1$. The
null hypothesis $H_{0i}$ is $\mu_i \le 0$, so the global null
hypothesis is $H_0:\,\mu_i \le 0,\,\forall i$. The individual
$p$-values are given by $p_i = 1 - \Phi(Y_i)$ where $\Phi$ is the CDF
of the standard normal distribution. The $p$-value $p_i$ has a uniform
distribution if $\mu_i = 0$ and $p_i$ is conservative if $\mu_i < 0$.

When some $p$-values are conservative, the classical methods in
\Cref{sec:previous-methods} usually lose power. Consider a
rather extreme example in which the $p$-values are generated from the
one-sided normal means problem described in the last paragraph with $n
= 100$, $\mu_1=\mu_2=3$ and $\mu_3=\dotsb=\mu_{100}= -10$. For
simplicity, let's say the observed statistics are just $Y_i = \mu_i$, so the
$p$-values are $(p_1,p_2,p_3\dotsc,p_{100}) \approx (0.001, 0.001, 1, \dotsc,
1)$. Although the first two $p$-values are highly significant, the
classical methods do not find the whole set of $p$-values providing
enough evidence to reject the global null hypothesis, because they
``over-correct'' for multiplicity. This is demonstrated in the first
row of \Cref{tab:hypothetical}.
In contrast, using the conditional $p$-values proposed in
\Cref{sec:conditional-tests}, the same tests all reject the global
null hypothesis.

\begin{table}[t]
  \centering
  \begin{tabular}{|l|rrr|}
    \hline
    & Bonferroni & Fisher & Truncated Product \\
    \hline
    Classical & $0.1$ & $1$ & $0.999$ \\
    Conditional ($\tau = 0.5$) & $0.004$ & $5.4 \times 10^{-5}$ & $4.72 \times 10^{-5}$ \\
    \hline
  \end{tabular}
  \caption{Combined $p$-values in the
    hypothetical example in \Cref{sec:probl-cons-tests} that $(p_1,p_2,p_3\dotsc,p_{100}) = (0.001, 0.001, 1, \dotsc,
    1)$. 
  }
  \label{tab:hypothetical}
\end{table}

Intuitively, if we do observe $(p_1,p_2,p_3\dotsc,p_{100}) = (0.001, 0.001, 1, \dotsc,
1)$, the first thing to be noticed is there are exceptionally many
large $p$-values. This indicates many conservative tests. Naturally,
we would like to ``ignore'' these large
$p$-values and only use the two smaller ones, with which we can easily
reject the global null. However, we cannot simply remove the large
$p$-values because this would be data snooping and make the
subsequent inference invalid.

The truncated product method of \citet{zaykin2002truncated} attempts
to address this problem by only multiplying the $p$-values below
some threshold $\tau$. This improves Fisher's combination test when
some $p$-values are conservative. However, it does not completely
resolve the problem, because the null distribution is still computed
assuming all $p$-values are uniformly distributed even though we have
overwhelming evidence that many of them are conservative.


\section{Conditional test}
\label{sec:conditional-tests}

\subsection{Testing the global null}
\label{sec:testing-global-null}

In light of the discussion above, we propose a simple conditional
test. Given independent $p$-values $p_1,p_2,\dotsc,p_n$ and a fixed
threshold parameter $0 < \tau \le 1$, let $\mathcal{S}_{\tau} = \{i,\,p_i \le
\tau\}$ be the indices where the $p$-values are less than $\tau$. When
$p_i$ is exact, it is uniformly distributed on $[0, \tau]$ given
$p_i \le \tau$. In other words, the conditional distribution $p_i/\tau
| S_{\tau} \overset{i.i.d.}{\sim} \mathrm{Unif}[0,1]$ if $p_i$ is exact.

Our proposal is to use
any of the global testing methods in \Cref{sec:previous-methods} on
the set of $p$-values $\{p_i/\tau,\,i \in \mathcal{S}_{\tau}
\}$ (we assume the combined $p$-value is $1$ if $\mathcal{S}_{\tau}$
is empty). Intuitively, this screens out the very large $p$-values in the
example in the last section. For example, the conditional Bonferroni test rejects the global null
$H_0$ if the cardinality of $\mathcal{S}_{\tau}$ is positive,
$|\mathcal{S}_{\tau}| > 0$, and
\begin{equation} \label{eq:p-cb}
  p^{\mathrm{CB}}(p_1,\dotsc,p_n;\tau) = \Big(\min_{1 \le i \le n} p_i/\tau\Big)
  \cdot |\mathcal{S}_{\tau}| \le \alpha.
\end{equation}
Notice that $\tau = 1$ reduces to the original Bonferroni
correction since $|S_1|=n$. The conditional Fisher test uses the
statistic
\[
  T^{\mathrm{CF}}(p_1,\dotsc,p_n;\tau) = -2 \sum_{i=1}^n [\log (p_i / \tau)] \cdot
  1_{\{p_i \le \tau\}}
\]
and compares it with the $\chi^2$ distribution with
$|\mathcal{S}_{\tau}|$ degrees of freedom (if $|\mathcal{S}_{\tau}| >
0$); denote the combined $p$-value by
$p^{CF}(p_1,\dotsc,p_n;\tau)$. Notice that this test
statistic is very similar to the truncated product of
\citet{zaykin2002truncated} except we also divide the $p$-values by
the threshold $\tau$. This allows us to work with $|\mathcal{S}_{\tau}|$
instead of $n$ many $p$-values. When many $p$-values are conservative,
$|\mathcal{S_{\tau}}|$ can be substantially smaller than $n\tau$, the
expected value of $\mathcal{|S_{\tau}|}$ when no $p$-value is
conservative. The simulation and real data examples in
Sections \ref{sec:simulation} and \ref{sec:meta-analysis-1} show the
conditional tests can
be much more powerful than the unconditional tests ($\tau = 1$).

However, the conditional tests are not valid without making further
assumptions. Heuristically, we need the transformed $p$-values
$\{p_i/\tau,\,i \in \mathcal{S}_{\tau}\}$ to be valid. Although
the transformed $p$-values are exact if the original $p$-values
are exact, the same conclusion does not in general hold for conservative
$p$-values. Next we introduce a stronger notion of conservativeness:
\begin{definition} \label{def:unif-cons}
  A valid $p$-value $p_i$ is called \emph{uniformly valid} if for all $0 <
  \tau < 1$ such that $\mathrm{P}(p_i \le \tau) > 0$, $p_i/\tau$ given
  $p_i \le \tau$ is valid. A $p$-value is called \emph{uniformly
    conservative} if it is conservative and uniformly valid.
\end{definition}
\begin{proposition} \label{prop:cond-test-valid}
  The conditional test with any fixed $0 < \tau \le 1$ (any global test on $\{p_i/\tau,i
  \in \mathcal{S}_{\tau}\}$) controls type I error at the nominal
  level if $p_1,p_2,\dotsc,p_n$ are independent and uniformly valid.
\end{proposition}

The proof of \Cref{prop:cond-test-valid} immediately follows from
\Cref{def:unif-cons} and the validity of the global test. Next we
examine which tests are uniformly valid/conservative. Let $F_i$ be
the CDF of $p_i$, $F_i(x) = \mathrm{P}(p_i \le x)$. By
\Cref{def:unif-cons}, the $p_i$ is
uniformly conservative if and only if
\[
F_i(\tau x) \le x F_i(\tau),~\forall 0 \le x, \tau \le 1.
\]
Geometrically, this means that the function $F_i(x)$ is always below
the segment from $(0,0=F_i(0))$ to $(\tau,F_i(\tau))$ if $0 \le x \le
\tau$. Therefore, a sufficient condition for uniform conservativeness
is convexity of the CDF, since by convexity,
\[
F_i(\tau x) = F_i((1 - x) \cdot 0 + x \cdot \tau) \le (1-x) F_i(0) + x
F_i(\tau) = x F_i(\tau).
\]
However, as pointed out by an anonymous reviewer, convexity is not
necessary for uniform conservativeness. The geometric interpretation of
uniform conservativeness and an example of nonconvex but uniformly
conservative CDF are illustrated in \Cref{fig:cdf}.

\begin{figure}[t]
  \centering
  \includegraphics[width = \textwidth]{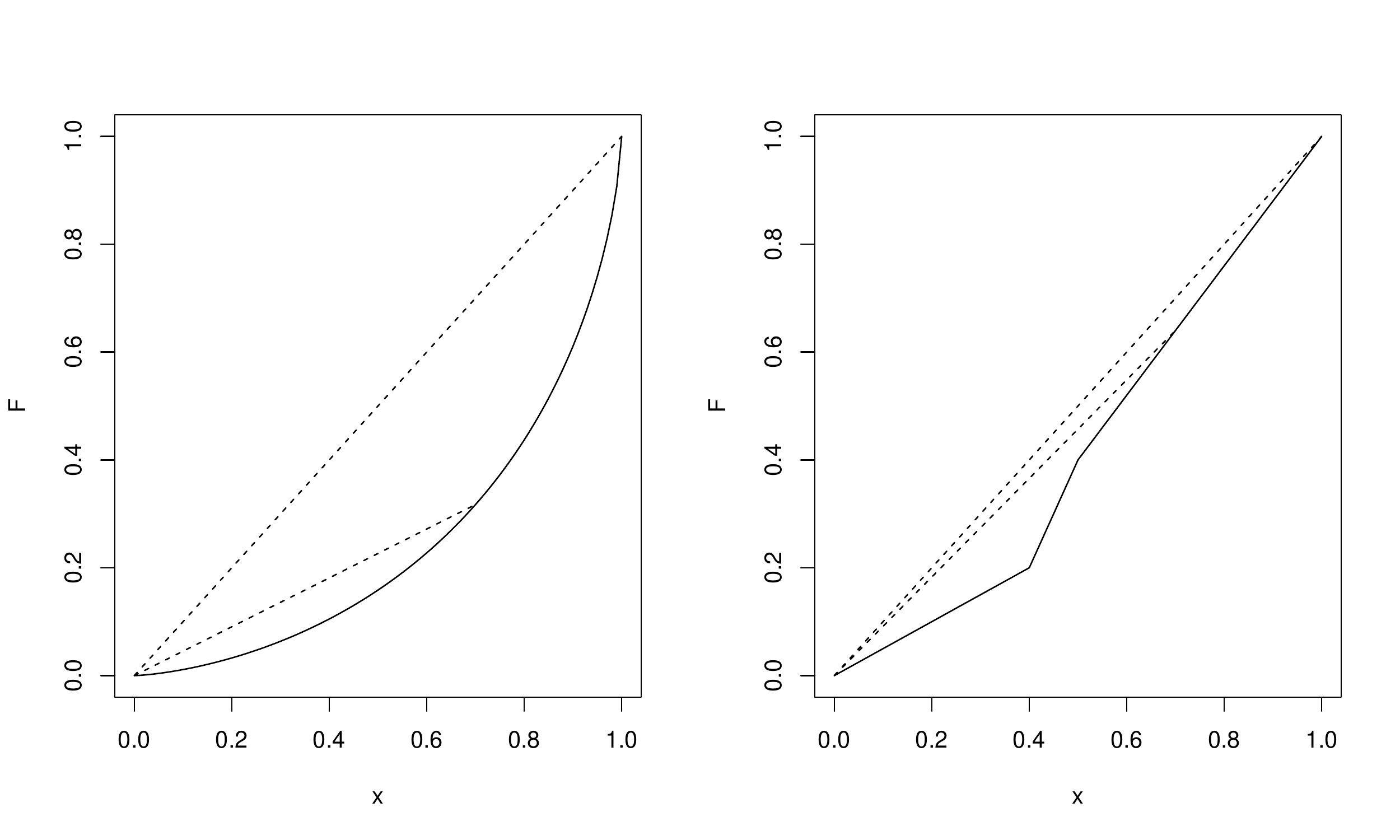}
  \caption{Two examples of uniformly conservative CDFs. The left plot
    is the distribution of $\Phi(Y)$ where $Y \sim
    \mathrm{N}(-1,1)$. The right plot corresponds to a piecewise
    constant density function: $f(x) = 0.5 \cdot I(0 \le x \le 0.4)
    + 2 \cdot I(0.4 < x \le 0.5) + 1.2 \cdot I(0.5 < x \le 1)$. Both
    CDFs satisfy the condition $F(x \tau) \le x F(\tau)$ for all $0 \le
    x,\tau \le 1$ so they are uniformly conservative. The geometric
    interpretation of this condition is illustrated by the two dashed lines corresponding to $\tau = 0.7$ and $1$. The right plot
    suggests that convexity of CDF is not necessary for uniform
    conservativeness.}
  \label{fig:cdf}
\end{figure}

When the CDF $F(x)$ is differentiable, convexity of $F(x)$ is
equivalent to the density $f(x)$ being monotonically increasing. This situation
arises when we are performing one-sided tests in a one-dimensional
exponential family, $\{g_{\theta}, -\infty < \theta <
\infty\}$. Suppose the sufficient statistic is $T$ (without loss of
generality we assume the mean of $T$ is increasing in $\theta$) and the null hypothesis is $H_0:\theta \le \theta_0$ vs.\ $H_1:\theta >
\theta_0$. The classical
Karlin-Rubin theorem states that the uniformly most powerful (UMP) test
rejects $H_0$ when $T$ is large, and a $p$-value can be computed using
the right-tail of $g_{\theta_0}$ (let $G_{\theta}$ be the CDF of
$g_{\theta}$) by $p = 1 - G_{\theta_0}(T)$. Therefore, if $T \sim g_{\theta}$,
\[
F(x) = \mathrm{P}(p \le x) = \mathrm{P}(T \ge G^{-1}_{\theta_0}(1 -
x)) = 1 - G_{\theta}(G^{-1}_{\theta_0}(1 -
x)).
\]
By the inverse function theorem, this implies that
\[
f(x) = \frac{g_{\theta}(G^{-1}_{\theta_0}(1-x))}{g_{\theta_0}(G^{-1}_{\theta_0}(1-x))}.
\]
It is well known that the exponential family has monotone
likelihood ratio (MLR), therefore $f(x)$ is increasing in $x$ if $\theta<
\theta_0$ (since $G^{-1}_{\theta_0}(1-x)$ is decreasing in $x$).

In summary, we have proved that
\begin{proposition} \label{prop:exp-unif-cons}
  When the true $\theta < \theta_0$, the UMP one-sided test of
  $H_0:\theta \le \theta_0$ vs.\ $H_1:\theta > \theta_0$ in the
  one-dimensional exponential family is uniformly conservative.
\end{proposition}

Our \Cref{prop:exp-unif-cons} can be viewed as a special case of
\citet[Theorem 1.1]{whitt1980uniform} who introduced a more general concept called uniform conditional stochastic order
(UCSO). When the sample space is totally ordered,
\citet{whitt1980uniform} showed that MLR implies UCSO (uniform
conservativeness). We refer the reader to \citet{whitt1980uniform} for
the more general result.




\subsection{Testing qualitative interaction}
\label{sec:power-test-qual}

The qualitative interaction problem formulated in
\Cref{sec:overview-paper} is a special case of one-sided testing in an
exponential family.  The $i$-th study ($1 \le i \le n$) provides an
effect estimate $X_i \sim \mathrm{N}(\mu_i, \sigma_i^2)$ with known
$\sigma_i$ and the null hypothesis of no qualitative interaction can be
separated into two global testing problems: all the means are
non-negative ($H_0^{+}$) and all the means are non-positive
($H_0^{-}$). Consider the first global null hypothesis $H_0^{+} =
\cap_{1 \le i \le n} H_{0i}^{+},\,H^{+}_{0i}:\mu_i \ge 0$. Since the
variance $\sigma_i^2$ is known, $H_{0i}^{+}$ is a one-sided problem in the
normal location family. By
\Cref{prop:cond-test-valid} and \Cref{prop:exp-unif-cons}, the conditional test
of $H_0^{+}$ is valid with any $0 < \tau \le 1$.

Since $H_0$ is the union of $H_0^{+}$ and $H_0^{-}$, we can reject $H_0$
at level $\alpha$ if both $H_0^{+}$ and $H_0^{-}$ are rejected at level
$\alpha$, because if $H_0^{+}$ is true,
\[
\mathrm{P}(H_0\mathrm{~is~rejected}) =
\mathrm{P}(H_0^{+}~\mathrm{and}~H_0^{-}\mathrm{~are~rejected}) \le
\mathrm{P}(H_0^{+}\mathrm{~is~rejected}) \le \alpha.
\]
Similarly, the type I error is also controlled if $H_0^{-}$ is true. This translates into the following testing procedure: for
$H_0^{+}$ and $H_0^{-}$, we can compute a combined $p$-value using
the global tests in \Cref{sec:procedure} (conditionally or
unconditionally). Then we report a single $p$-value for $H_0$
using the larger of the two combined $p$-values and reject $H_0$ if it
is less than the significance level $\alpha$.

In \Cref{sec:power-test-qual}, we will compare the performance of our
method with two existing tests of qualitative interaction that are
widely used in practice. The first method is the likelihood ratio test
(LRT) of \citet{gail1985testing}. The second method is the interval based
graphical approach (IBGA) of \citet{pan1997test}, which is equivalent
to the procedure described above using \citet{vsidak1967rectangular}'s
correction as the global test (applied to the unconditional $p$-values).




\subsection{Adaptively selecting the threshold $\tau$}
\label{sec:select-thresh-tau}

A remaining practical issue is how to choose the threshold $\tau$. Here we
provide an adaptive strategy that attempts to select $\tau$ without
sacrificing the validity of the test. Our strategy is
based on the following observation
\begin{proposition} \label{prop:stopping-time}
  If $\tau$ is a backward stopping time in the sense that $\{\tau \ge x\}
  \independent \{p_i,\,i\in\mathcal{S}_x\}$ for any $0 \le x \le 1$, then
  \Cref{prop:cond-test-valid} still holds.
\end{proposition}
This Proposition is true because our conditional test is based on
$\{p_i/\tau,\,i\in\mathcal{S}_{\tau}\}$, which is independent of how $\tau$
is selected if $\tau$ is a backward stopping time.

\Cref{prop:stopping-time} suggests an interactive strategy to choose
$\tau$:
\begin{enumerate}
\item The data analyst chooses a sequence of decreasing cutoffs,
  $\tau_1$, $\tau_2$, $\tau_3$, \ldots, $\tau_K$ (for example, $0.9$, $0.85$,
  $0.8$, \ldots, $0.1$).
\item At step $k \ge 1$, the data analyst decides if she wants to
  continue based on $\{p_i,\,i\not\in\mathcal{S}_{\tau_k}\} =
  \{p_i,\,p_i > \tau_k\}$. Denote the $\tau_k$ she stops at as $\tau$.
\item Apply a global test on $\{p_i/\tau,\,i\in\mathcal{S}_{\tau}\}$.
\end{enumerate}

In principle, the starting cutoff $\tau_1$ should not be too close to
$1$, otherwise there is little information for the data analyst to
decide if she wants to move on. The ending cutoff $\tau_K$ should not
be too close to $0$, so not too many signals are excluded.

Finally, we describe when the data analyst may want to stop. Consider
the conditional Bonferroni test defined in \eqref{eq:p-cb}. Since the
minimum $p$-value does not depend on the threshold $\tau$ (unless
$\tau$ is very small so $\mathcal{S}_{\tau}$ is empty), it is
reasonable to devise an adaptive strategy
to minimize $|\mathcal{S}_{\tau}|/\tau$. Let $\mathcal{F}$ be the
distribution of the $p$-values: $F(x) = (1/n)\sum_{i=1}^n F_i(x)$,
then $|\mathcal{S}_{\tau}|/\tau \approx [nF(\tau)]/\tau$. Notice that
\[
\frac{\diff}{\diff \tau} \frac{F(\tau)}{\tau} = \frac{f(\tau)\tau -
  F(\tau)}{\tau^2}.
\]
Therefore, a sensible criterion is to stop at step $k$ if there is no
strong evidence that $f(\tau_k) \tau_k - F(\tau_k) > 0$. More
specifically, let $0 < w \le 1 - \tau_1$ be some prespecified window
size (for example, $0.1$). We can estimate $F$ and $f$ by
\[
\hat{F}(\tau) = \frac{|\mathcal{S}_{\tau}|}{n},~\hat{f}(\tau) =
\frac{|\{i,\,\tau \le p_i \le \tau + w\}|}{nw},
\]
and stop if we fail to reject $nw \hat{f}(\tau_k) \sim
\mathrm{Binomial}(n,qw)$ with $q < \hat{F}(\tau_k)/\tau_k$ at some prespecified
significance level (for example, $1\%$). We implement this heuristic
strategy in the numerical studies in
\Cref{sec:simulation,sec:meta-analysis-1} and find it generally
improves the power of the conditional tests with fixed $\tau$. It also
works well with other global tests too.

\subsection{Beyond global testing}
\label{sec:beyond-glob-test}

So far we have focused on testing the global null hypothesis that all
the individual hypotheses are true. When the global null
is rejected, it is often interesting to know which individual
hypotheses are false. In this case, it is often desirable to control
some multiple testing criterion such as the family-wise error rate (FWER) and
the false discovery rate (FDR).

The conditional tests proposed above can be easily extended to general
multiple testing problems. In fact, the conditional tests are closely
related to the selective inference framework of \citet{fithian2014optimal} by viewing
$\mathcal{S}_{\tau}$ as model selection. \citet{fithian2014optimal}
argued that the statistical inference should be performed conditioning
on the selection event $\mathcal{S}_{\tau}$. See
\citet{benjamini2010simultaneous} for a discussion on the difference
between simultaneous and selective inference. Notice
that the conditional $p$-values $\{p_i/\tau,\,i\in\mathcal{S}_{\tau}\}$ can be viewed
as usual $p$-values. We can apply, for example, Hochberg's step-up
procedure to control the FWER, or
the Benjamini-Hochberg procedure to control the FDR.
In general, we expect the procedures using conditional p-values will
be more powerful than their unconditional versions when many tests are
conservative.

\section{Some theoretical results}
\label{sec:theory}
\newcommand{\Cov}{\operatorname{\textrm{Cov}}}
\newcommand{\Var}{\operatorname{\textrm{Var}}}
\newcommand{\goto}{\rightarrow}
\renewcommand{\P}{\operatorname{\mathbb{P}}}
\newcommand{\E}{\operatorname{\mathbb{E}}}
\newcommand{\e}{\mathrm{e}}


\subsection{Power of the conditional test}
\label{sec:power-cond-test}


\begin{theorem}\label{thm:power}
Suppose all the $p$-values are independent and uniformly valid, and the
cutoff $0 < \tau < 1$ is a fixed constant. Let $F_i$ be the CDF of
the $i$-th $p$-value.
Then the Bonferroni-adjusted conditional $p$-value
$p^{\mathrm{CB}} = (\min_{1 \le i \le n}
p_i/\tau) |\mathcal{S}_{\tau}|$ and the Bonferroni-adjusted unconditional $p$-value
$p^{\mathrm{B}} = n \min_{1 \le i \le n} p_i$ satisfy
\begin{equation} \label{eq:thm-bonf}
\liminf_{n \to \infty} \frac{1}{\tau n} \sum_{i=1}^n
F_i(\tau) + o_p(1) \le \liminf_{n \to \infty}
\frac{p^{\mathrm{CB}}}{p^{\mathrm{B}}} \le \limsup_{n \to \infty} \frac{p^{\mathrm{CB}}}{p^{\mathrm{B}}} \le \limsup_{n \to \infty} \frac{1}{\tau n} \sum_{i=1}^n F_i(\tau) + o_p(1).
\end{equation}
Therefore, if the right hand side of \eqref{eq:thm-bonf} is less than
$1$, the conditional Bonferroni test is asymptotically more powerful
than the conventional Bonferroni test.
\end{theorem}

Consider the case that the number of
non-null $p$-values is a vanishing
fraction of $n$ (the situation in which the Bonferroni test is
desirable; see, for example, \citet{arias2011global}). Recall that
a $p$-value $p_i$ is valid if $F_i(\tau) \le \tau$ for all $\tau$ and is
conservative if it is valid and the inequality $F_i(\tau) \le \tau$ is
strict for some $\tau$. Therefore, the right hand side of
\eqref{eq:thm-bonf} is always not
greater than $1$. Furthermore, when the non-nulls are sparse and the
fraction of nulls that are conservative at $\tau$, $|\{1 \le i \le
n:\,F_i(\tau) <   \tau\}|/n$, is non-negligible, the right hand side
of \eqref{eq:thm-bonf} is less than $1$ and thus the conditional
Bonferroni test is more powerful than the unconditional test (consider
the example in \Cref{sec:probl-cons-tests}).



\subsection{Validity of the conditional test under dependence}
\label{sec:valid-cond-test}

Next, we show that the conditional test is asymptotically valid when
the test statistics exhibit exchangeable correlations.
Suppose $(Y_1, \ldots, Y_n)$ follow the multivariate normal
distribution with zero mean, unit variance and equal covariance $\rho$. The
correlation $\rho$ can arbitrarily vary as $n$ increases, with the
exception that $\rho$ is bounded away from 1. For example,
$\rho$ can be any number no greater than $0.99$. Note that, although
$\rho$ can be negative, it obeys $\rho \ge -\frac1{n-1}$ in order to
keep the covariance of $Y_1, \ldots, Y_n$ positive
semidefinite. Finally, let $p_i = 1 - \Phi(Y_i)$ be the one-sided
$p$-value.

The next Theorem states that the conditional Bonferroni test still
controls the type I error asymptotically when the test statistics
$Y_i$ are not independent but equally correlated.

\begin{theorem}\label{thm:equi-corr}
In the above setting, we have
$
\mathrm{P}\left((\min_{1 \le i \le n}
p_i/\tau) |\mathcal{S}_{\tau}| \le \alpha \right) \le (1 + o(1)) \alpha.
$
\end{theorem}

\section{Simulation}
\label{sec:simulation}

\subsection{Power of the global test}
\label{sec:global-test}

To assess the performance of the proposed procedures in
\Cref{sec:conditional-tests}, we implement a simulation study with
$n = 100$ one-sided tests of normal means $H_{0j}:\mu_j \le 0$ vs.\
$H_{1j}:\mu_j > 0$ in the following five
settings ($\mu_{i:j}$ stands for the vector $(\mu_i,\dotsc,\mu_j)$):
\begin{enumerate}
\item All null: $\mu_{1:100} = 0$;
\item 1 strong 99 null: $\mu_1 = 4$, $\mu_{2:100} = 0$;
\item 1 strong 99 conservative: $\mu_1 = 4$, $\mu_{2:100} =
  -1$;
\item 20 weak 80 null: $\mu_{1:20} = 1$, $\mu_{21:100} = 0$;
\item 20 weak 80 conservative: $\mu_{1:20} = 1$, $\mu_{21:100} = -1$.
\end{enumerate}
The test statistics are generated by $Y_i \sim \mathrm{N}(\mu_i,1)$
and the $p$-values are computed by $p_i = 1 - \Phi(Y_i)$ where $\Phi$
is the CDF of standard normal. Then we apply four combination tests, Bonferroni, Fisher, Tukey and truncated product, in two forms---the original
unconditional test in \Cref{sec:previous-methods} and the conditional
($\tau = 0.5$). The null distribution of Tukey's test is
approximated by $10000$ samples from the global null.

\Cref{tab:simulation-power} reports the power of these tests in
$10000$ simulations when the significance level is $\alpha =
0.05$. When all the null hypotheses are true, all the tests have power
(a.k.a.\ type I error) about $5\%$. When there is only one strong
signal, Bonferroni performs the best, but notice that conditioning
does not make the power deteriorate substantially when no test is conservative (in
fact it improves the power of Fisher, Tukey and truncated product) and
substantially improves the power when many tests are conservative. The
same thing is true when we have many weak signals, except that
Fisher's combination and the truncated product perform the best in
this regime.

\begin{table}
  \centering
\begin{tabular}{|ll|rrr|}
\hline
Setting & Method & Uncond. & Cond. & Adaptive \\
\hline
1. All null & Bonferroni  & $\phantom{0}4.9$ & $\phantom{0}4.9$ & $\phantom{0}4.9$ \\
 & Fisher  & $\phantom{0}5.1$ & $\phantom{0}4.8$ & $\phantom{0}5.1$ \\
 & Tukey  & $\phantom{0}5.3$ & $\phantom{0}5.0$ & $\phantom{0}5.4$ \\
 & TruncatedP  & $\phantom{0}5.0$ & $\phantom{0}4.9$ & $\phantom{0}5.0$ \\
\hline
2.\ 1 strong 99 null & Bonferroni  & $\mathbf{78.0}$ & $\mathbf{78.0}$ & $\mathbf{78.0}$ \\
 & Fisher  & $25.9$ & $34.7$ & $27.2$ \\
 & Tukey  & $\phantom{0}7.0$ & $\phantom{0}6.9$ & $\phantom{0}7.5$ \\
 & TruncatedP  & $23.4$ & $31.2$ & $24.5$ \\
\hline
3.\ 1 strong 99 conservative & Bonferroni  & $\mathbf{76.2}$ & $\mathbf{85.1}$ & $\mathbf{88.7}$ \\
 & Fisher  & $\phantom{0}0.0$ & $20.3$ & $84.7$ \\
 & Tukey  & $\phantom{0}0.0$ & $\phantom{0}0.1$ & $\phantom{0}4.2$ \\
 & TruncatedP  & $\phantom{0}0.0$ & $21.0$ & $84.4$ \\
\hline
4.\ 20 weak 80 null & Bonferroni  & $22.8$ & $20.5$ & $22.3$ \\
 & Fisher  & $\mathbf{73.9}$ & $\mathbf{57.2}$ & $\mathbf{71.4}$ \\
 & Tukey  & $57.8$ & $40.4$ & $54.4$ \\
 & TruncatedP  & $70.2$ & $53.9$ & $67.2$ \\
\hline
5.\ 20 weak 80 conservative & Bonferroni  & $\mathbf{20.0}$ & $28.2$ & $28.1$ \\
 & Fisher  & $\phantom{0}0.0$ & $48.7$ & $\mathbf{52.3}$ \\
 & Tukey  & $\phantom{0}0.5$ & $25.9$ & $30.8$ \\
 & TruncatedP  & $\phantom{0}0.3$ & $\mathbf{51.0}$ & $\mathbf{52.9}$ \\
\hline
\end{tabular}
  \caption{Power (in \%) in the $5$ simulation settings in which the
    global tests (Bonferroni, Fisher, Tukey, and TruncatedP) are
    applied to the unconditional and conditional $p$-values. The
    truncation threshold $\tau$ is $0.5$ or chosen adaptively as
    described in \Cref{sec:select-thresh-tau}.}
  \label{tab:simulation-power}
\end{table}


\subsection{Power of signal detection}
\label{sec:bonf-corr-contr}

Next we investigate the empirical performance of four Bonferroni procedures that
control the family-wise error rate (FWER): the original Bonferroni
correction and the conditional Bonferroni test with $\tau = 0.5$ and
$\tau = 0.8$. 
We simulate $n = 1000$ one-sided tests of normal means in
the following two settings:
\begin{enumerate}
\item No conservative: $\mu_{1:20} = 4$, $\mu_{21:1000} = 0$;
\item Conservative: $\mu_{1:20} = 4$, $\mu_{21:1000} = -1$.
\end{enumerate}
The test statistics are generated by $Y_i \sim \mathrm{N}(\mu_i,1)$
and the $p$-values are computed by $p_i = 1 - \Phi(Y_i)$.

\begin{table}
  \centering
  \begin{tabular}{|l|l|r|r|r|r|}
    \hline
    Setting & Method & 1st Quart. & Med. & Mean & 3rd Quart. \\
    \hline
    \multirow{3}{*}{No conservative} & Bonferroni & 9 & 11 & 10.91 & 12 \\
            & Cond.\ Bonf.\ ($\tau = 0.5$) & 9 & 11 & 10.87 & 12 \\
            & Cond.\ Bonf.\ ($\tau = 0.8$) & 9 & 11 & 10.90 & 12 \\
            & Cond.\ Bonf.\ (adaptive $\tau$) & 10 & 11 & 10.90 & 12 \\
    \hline
    \multirow{3}{*}{Conservative} & Bonferroni & 9 & 11 & 10.78 & 12 \\
            & Cond.\ Bonf.\ ($\tau = 0.5$) & 11 & 13 & 12.78 & 14 \\
            & Cond.\ Bonf.\ ($\tau = 0.8$) & 10 & 12 & 11.88 & 13 \\
            & Cond.\ Bonf.\ (adaptive $\tau$) & 12 & 13 & 13.12 & 15 \\
    \hline
  \end{tabular}
  \caption{Quartiles and means of the number of correct rejections in
    $1000$ simulations.}
  \label{tab:multiple}
\end{table}

\Cref{tab:multiple} compares the power of four different
Bonferroni procedures. The numbers in \Cref{tab:multiple} are summary
statistics of the number of correct rejections in $1000$
simulations. For example, when no test is conservative, the original
Bonferroni procedure rejects $10.91$ tests on average, while the conditional
Bonferroni procedures reject $10.87$, $10.90$ and $10.90$ tests on
average depending on which threshold $\tau$ is used. In contrast, in
the conservative scenario the conditional Bonferroni with
$\tau = 0.5$ rejects about $2$ more tests than the original
Bonferroni procedure on average. The conditional Bonferroni's method
with adaptively selected $\tau$ makes even more discoveries. In
conclusion, the proposed conditional tests are slightly less powerful
when no test is conservative, but are much more powerful when many
tests are conservative.

\subsection{Power of testing qualitative interaction}
\label{sec:power-test-qual}

Finally, we study the power of the tests of qualitative interaction
that are described in \Cref{sec:power-test-qual}. We simulate $n = 100$ normal
variables $Y_i \sim \mathrm{N}(\mu_i,1)$ in the following six settings:
\begin{enumerate}
\item 1 positive 99 null: $\mu_1 = 4$, $\mu_{2:100} = 0$.
\item 1 positive 1 negative: $\mu_1 = 4$, $\mu_2 = -4$, $\mu_{3:100} =
  0$;
\item 1 positive 99 negative: $\mu_1 = 4$, $\mu_{2:100} = -1$;
\item 20 positive 80 negative: $\mu_{1:20}=1$, $\mu_{21:100} = -1$;
\item 50 positive 50 negative: $\mu_{1:50}=1$, $\mu_{51:100} = -1$;
\item Gradual (1st setting): $\mu_{1:100}$ are equally spaced between
  $-1.5$ and $2$;
\item Gradual (2nd setting): $\mu_{1:100}$ are equally spaced between
  $-1.5$ and $4$.
\end{enumerate}

Apart from the first setting, the null hypothesis of no qualitative
interaction is false. \Cref{tab:simulation-power-qualint} compares the power of
the proposed tests with two existing methods, the likelihood ratio test (LRT) of
\citet{gail1985testing} and the interval based graphical approach
(IBGA) of \citet{pan1997test} as implemented in the \texttt{R} package
\texttt{QualInt} \citep{QuanInt}. The performance of IBGA is very
similar to the unconditional Bonferroni test, since IBGA is equivalent
to applying {\v{S}}id{\'a}k's correction to the original $p$-values in
our framework and it is well known that {\v{S}}id{\'a}k's correction
is only slightly more powerful than Bonferroni's correction.

\begin{table}
  \centering \small
\begin{tabular}{|ll|rrr|}
\hline
Setting & Method & Uncond. & Cond. & Adaptive \\
\hline
1 positive 99 null & Bonferroni  & $\phantom{0}3.6$ & $\phantom{0}3.6$ & $\phantom{0}3.6$ \\
 & Fisher  & $\phantom{0}0.1$ & $\phantom{0}1.7$ & $\phantom{0}1.0$ \\
 & Tukey  & $\phantom{0}0.0$ & $\phantom{0}0.3$ & $\phantom{0}0.3$ \\
 & TruncatedP  & $\phantom{0}0.5$ & $\phantom{0}1.6$ & $\phantom{0}1.4$ \\
 & IBGA  & $\phantom{0}3.7$ & $\phantom{0}3.7$ & $\phantom{0}3.7$ \\
 & LRT  & $\phantom{0}1.2$ & $\phantom{0}1.2$ & $\phantom{0}1.2$ \\
1 positive 1 negative & Bonferroni  & $59.9$ & $\mathbf{59.9}$ & $\mathbf{60.0}$ \\
 & Fisher  & $\phantom{0}1.0$ & $11.6$ & $\phantom{0}6.1$ \\
 & Tukey  & $\phantom{0}0.0$ & $\phantom{0}0.4$ & $\phantom{0}0.3$ \\
 & TruncatedP  & $\phantom{0}2.9$ & $\phantom{0}9.7$ & $\phantom{0}6.5$ \\
 & IBGA  & $\mathbf{60.4}$ &  &  \\
 & LRT  & $12.8$ &  &  \\
1 positive 99 negative & Bonferroni  & $50.9$ & $\mathbf{45.4}$ & $57.6$ \\
 & Fisher  & $\phantom{0}0.0$ & $19.6$ & $\mathbf{84.9}$ \\
 & Tukey  & $\phantom{0}0.0$ & $\phantom{0}0.2$ & $\phantom{0}4.2$ \\
 & TruncatedP  & $\phantom{0}0.0$ & $20.7$ & $\mathbf{84.5}$ \\
 & IBGA  & $\mathbf{51.9}$ &  &  \\
 & LRT  & $\phantom{0}0.0$ &  &  \\
20 positive 80 negative & Bonferroni  & $\mathbf{11.7}$ & $14.4$ & $16.7$ \\
 & Fisher  & $\phantom{0}0.0$ & $49.6$ & $\mathbf{51.8}$ \\
 & Tukey  & $\phantom{0}0.6$ & $27.4$ & $30.0$ \\
 & TruncatedP  & $\phantom{0}0.3$ & $\mathbf{51.4}$ & $\mathbf{52.5}$ \\
 & IBGA  & $\mathbf{12.1}$ &  &  \\
 & LRT  & $\phantom{0}3.0$ &  &  \\
50 positive 50 negative & Bonferroni  & $18.6$ & $18.7$ & $21.7$ \\
 & Fisher  & $71.5$ & $\mathbf{97.1}$ & $\mathbf{98.3}$ \\
 & Tukey  & $73.7$ & $90.7$ & $94.0$ \\
 & TruncatedP  & $92.5$ & $94.9$ & $\mathbf{98.1}$ \\
 & IBGA  & $19.5$ &  &  \\
 & LRT  & $\mathbf{93.8}$ &  &  \\
Gradual (from $-1.5$ to $2$) & Bonferroni  & $26.5$ & $28.0$ & $30.2$ \\
 & Fisher  & $18.3$ & $\mathbf{86.8}$ & $\mathbf{87.9}$ \\
 & Tukey  & $29.6$ & $71.2$ & $72.4$ \\
 & TruncatedP  & $53.7$ & $85.0$ & $\mathbf{88.2}$ \\
 & IBGA  & $27.4$ & &  \\
 & LRT  & $\mathbf{67.5}$ &  & \\
Gradual (from $-1.5$ to $4$) & Bonferroni  & $\mathbf{24.8}$ & $35.4$ & $36.4$ \\
 & Fisher  & $\phantom{0}0.0$ & $\mathbf{72.9}$ & $\mathbf{73.7}$ \\
 & Tukey  & $\phantom{0}1.0$ & $51.9$ & $48.9$ \\
 & TruncatedP  & $\phantom{0}1.1$ & $70.6$ & $\mathbf{73.6}$ \\
 & IBGA  & $\mathbf{25.3}$ &  &  \\
 & LRT  & $\phantom{0}7.7$ &  & \\
\hline
\end{tabular}
  \caption{Power (in \%) of testing qualitative interaction in the $7$
    simulation settings. The proposed methods---global tests
    (Bonferroni, Fisher, Tukey, and TruncatedP) applied to the
    unconditional and conditional $p$-values---are compared with two
    existing methods, the interval based graphical approach (IBGA)
    \citep{pan1997test} and the likelihood ratio test (LRT)
    \citep{gail1985testing}. For the conditional global tests,
    the truncation threshold $\tau$ is $0.5$ or chosen adaptively as
    described in \Cref{sec:select-thresh-tau}.}
  \label{tab:simulation-power-qualint}
\end{table}

Across
all settings and all global tests, conditioning (whether using $\tau =
0.5$ or $\tau$ adaptively chosen) improves the power of detecting
qualitative interaction. Apart from the second setting with only 1
positive and 1 negative signal where the Bonferroni/{\v{S}}id{\'a}k's
tests have the most power, in all the other settings the conditional
Fisher's test (and its variant, truncated product) with adaptively
chosen $\tau$ is always the most powerful method. In practice,
if there is qualitative interaction, it is rare that there is only one
subgroup with strong signal of the opposite sign (in other words,
the last four settings are more plausible than the second and third settings). Therefore, we expect the conditional Fisher's test with
adaptively chosen $\tau$ to perform the best in practice among the
tests considered in this paper.

\section{Applications to educational interventions}
\label{sec:meta-analysis-1}

\subsection{Random effects model only tests heterogeneity of treatment
effect}
\label{sec:random-effects-model}

In \Cref{sec:introduction} we introduced two motivating applications in
evaluating educational interventions. The standard practice to analyze
such datasets is the linear fixed/random/mixed effects model. For
example, for the modified school calendar intervention, a typical random
effects model is
\begin{equation} \label{eq:mixed-effect-model}
Y_i = \mu_i + \epsilon_i = (\mu + \alpha_{\!\mathsmaller{D_i}} + \beta_{\!\mathsmaller{S_i}}) + \epsilon_i,~i=1,\dotsc,n,
\end{equation}
where $Y_i$ is the observed average treatment effect in study $S_i$
nested in district $D_i$,
$\mu$ is the overall treatment effect, $\alpha_{\!\mathsmaller{D_i}} \sim
\mathrm{N}(0,\sigma_D^2)$ is the random district effect,
$\beta_{\!\mathsmaller{S_i}} \sim \mathrm{N}(0,\sigma_S^2) $ is the random study
effect (nested in the district), and $\epsilon_i \sim \mathrm{N}(0,
\sigma_i^2)$ is the noise with known sampling variance $\sigma_i^2$
because the per-study effect $Y_i$ is aggregated over many individuals.

To test heterogeneity of treatment effect, we use the function
\texttt{rma.mv} in the \texttt{R} package \texttt{metafor}
\citep{metafor} to fit a
multi-level random effects model using restricted maximum
likelihood. For $\sigma_D^2$, the point
estimate is $0.0651$ and the 95\% confidence interval is $(0.0222,
0.2072)$. For $\sigma_S^2$, the point estimate is $0.0327$ and the
95\% confidence interval is $(0.0163, 0.0628)$. Therefore, there is
strong evidence that the treatment effect of modified school calendar
varies across districts and schools. This is consistent with the
conclusions of \citet{konstantopoulos2011fixed}.

In general, the random effects model is not suitable for testing
qualitative interaction. If we take model
\eqref{eq:mixed-effect-model} and its random effects assumptions
literally, there is always a positive chance that some $\mu_i$ is
negative if $\sigma_D^2>0$ or $\sigma_S^2>0$. In other words, the
hypothesis of no qualitative interaction is automatically false in
a random effects model.

Alternatively, we may treat the district effects
$\alpha_{\mathrm{D}_i}$ as fixed and ask if any district experiences
a negative treatment effect. Outputs for this mixed effect model is
reported in \Cref{tab:modified-school-calendar}, where none of the
districts shows a significantly negative effect.

\begin{table}[t]
\centering
\caption{Output table of a mixed effect model for the modified school
  calendar application. In this model, each district has a fixed effect
  but the schools have random effects.}
\label{tab:modified-school-calendar}
\begin{tabular}{|l|rrrr|}
  \hline
 & Estimate & Std.\ Err.\ & $z$-value & $p$-value \\
\hline
  District 11 & -0.129 & 0.181 & -0.71 & 0.476 \\
  District 12 & 0.063 & 0.064 & 0.98 & 0.325 \\
  District 18 & 0.347 & 0.083 & 4.18 & 0.000 \\
  District 27 & 0.486 & 0.040 & 12.01 & 0.000 \\
  District 56 & 0.041 & 0.042 & 0.98 & 0.329 \\
  District 58 & -0.042 & 0.033 & -1.30 & 0.194 \\
  District 71 & 0.879 & 0.064 & 13.75 & 0.000 \\
  District 86 & -0.029 & 0.015 & -1.86 & 0.063 \\
  District 91 & 0.250 & 0.044 & 5.68 & 0.000 \\
  District 108 & 0.015 & 0.079 & 0.19 & 0.853 \\
  District 644 & 0.157 & 0.137 & 1.14 & 0.253 \\
   \hline
\end{tabular}
\end{table}

\subsection{Applying the proposed tests for qualitative interaction}
\label{sec:prop-tests-qual}

We apply the tests for qualitative interaction described in
\Cref{sec:power-test-qual} to the two datasets. The results are
reported in \Cref{tab:meta-analysis}.

\begin{table}[t]
  \centering
  \caption{Combined $p$-values for qualitative interaction in the
    motivating applications in \Cref{sec:motivating-examples}. Three
    versions of the Bonferroni's test and Fisher's combination test
    are used: the unconditional test (Unc.), the conditional test with
    threshold $0.5$ and $0.8$, and the conditional
    test with adaptively selected threshold $\tau$.}
  \label{tab:meta-analysis}
  \begin{tabular}{|ll|rrrr|}
    \hline
    & & Unc. & $\tau = 0.5$ & $\tau = 0.8$ & $\tau$ adaptive \\
    \hline
    Modified calendar (school) & Bonferroni & 0.044 & 0.031 & 0.034 &
    0.033 \\
    & Fisher & 0.224
             & $<0.001$ & 0.004 & 0.003 \\
    & IBGA & 0.042 & & & \\
    & LRT & 0.011 & & & \\
    \hline
    Modified calendar (district) & Bonferroni & 0.347 & 0.189 & 0.158 & 0.245 \\
    & Fisher & 0.788 &
                       0.113 & 0.088 & 0.374 \\
    & IBGA & 0.274 & & & \\
    & LRT & 0.351 & & & \\
    \hline
    Writing-to-learn & Bonferroni & 0.830 & 0.381 & 0.519 & 0.503 \\
    & Fisher & 1 & 0.578 & 0.917 & 0.877 \\
    & IBGA & 0.556 & & & \\
    & LRT & 0.985 & & & \\
    \hline
  \end{tabular}
\end{table}

For the modified school calendar intervention, none of the individual
schools has a strong
enough effect after Bonferroni's correction to reject the hypothesis at
significance level $0.01$. Conditioning ($\tau = 0.5$ and $\tau =
0.8$) helps to make the Bonferroni
adjusted $p$-values smaller, but they are still greater than
$0.01$. In contrast, by combining the weak evidence from several
schools and by reducing the number of conservative $p$-values via
conditioning, the conditional Fisher's test with $\tau = 0.5$ gives a
$p$-value of $0.0002$. The $p$-value
is still significant when the truncation threshold is set to $\tau =
0.8$. Without conditioning, Fisher's combination test does not have
enough power to detect the qualitative interaction in this
application.

We can also test if the treatment effect has qualitative interaction
among the districts. Using the $z$-values in
\Cref{tab:modified-school-calendar}, we apply the same global tests
and obtained $6$ $p$-values in the second row of
\Cref{tab:meta-analysis}. None of them is significant at level $0.05$,
indicating insufficient evidence of qualitative interaction in the
district level.

For the writing-to-learn intervention, all the tests cannot reject the null
hypothesis of no qualitative interaction. In the forest plot
(\Cref{fig:meta-analysis-2}), Ayers (1993) study shows a
significantly negative effect, but our results suggest that it is
plausible that is due to random chance. 

\section{Discussion}
\label{sec:conclusion}


\subsection{Multiparameter hypothesis testing}
\label{sec:mult-hypoth-test}

The global testing problem considered in this paper is also closely related
to the multiparameter hypothesis testing problem considered by
\citet{lehmann1952testing,berger1982multiparameter} that tests $H_0:
\theta \le 0$ for a multidimensional parameter $\theta$.
\citet{lehmann1952testing} showed that in general there is
no unbiased test for this problem, i.e.\ apart from the trivial test
that has constant power function, any valid test must have power less
than $\alpha$ at some alternative. In our paper, we assume there are
independent tests for each individual hypothesis $\theta_i \le 0$
and we restrict our attention to the alternative that many $\theta_i$s
are much smaller than $0$, so our results do not contradict the
conclusions in \citet{lehmann1952testing}. Another distinction is that we
allow the
dimension of $\theta$ to go to infinity, while in the classical
multiparameter setting the dimension of $\theta$ is fixed.

\subsection{Uniform validity/conservativeness}
\label{sec:uniform-validity}

As mentioned in \Cref{sec:unif-valid}, not all conservative tests are
uniformly conservative. One important exception we are aware of is the
sensitivity analysis of observational studies \citep[][Chapter
4]{rosenbaum2002observational} which places
bounds on the $p$-value for a specific magnitude of departure from
random treatment assignment. When there is no treatment effect,
the $p$-value under random treatment assignment is uniformly
distributed and not conservative, but the $p$-value bounds under
departure from randomization are inevitably very conservative because
many possible departures are considered.
Unfortunately, the $p$-value bounds are generally not uniformly
valid 
\citep[see e.g.,][]{zhao2017sensitivity}. 
When uniform conservativeness does
not hold, other methods (e.g.\
sample splitting in \citet{heller2009split}) must be used to reduce
the number of hypotheses. However, sample splitting loses some
efficiency because it discards some information in the
data whereas the conditional test proposed in this paper makes full
use of the information.

Another notable exception of uniform
validity is when the $p$-values are discrete. In this case,
the $p$-values cannot be strictly uniformly valid. However, for
tests with asymptotically normal approximations (such as the
bootstrap or Wilcoxon's rank sum test), we expect the conditional
tests in this paper are still asymptotically valid.

\appendix
\section{Proofs}
\label{sec:proofs}

\subsection{\Cref{prop:cond-test-valid}}
\begin{proof}
  Conditioning on the set $\mathcal{S}_{\tau}$, for any $i,j \in
  \mathcal{S}_{\tau}$, $p_i/\tau$ is a valid $p$-value and $p_i/\tau$
  is independent of $p_j/\tau$. Therefore, conditioning on the set
  $\mathcal{S}_{\tau}$, the global test on
  $\{p_i/\tau,i\in\mathcal{S}_{\tau}\}$ controls type I error at the
  nominal level ( on $\mathcal{S}_{\tau}$). By marginalizing over
  $\mathcal{S}_{\tau}$, the statement holds unconditionally as well.
\end{proof}



\subsection{Folded normal distribution}
\label{sec:fold-norm-distr}

\begin{proposition} \label{prop:folded-normal}
  The family of folded normal distributions with standard deviation
  $\sigma = 1$ and varying $\mu$ has monotone likelihood
  ratio. More precisely, if $\mu_1 > \mu_2 \ge 0$, then
  \begin{equation} \label{eq:foldnorm-mlr}
    \frac{\partial}{\partial x} \frac{\phi(x - \mu_1) + \phi(x +
      \mu_1)}{\phi(x - \mu_2) + \phi(x + \mu_2)} > 0,~\forall x > 0.
  \end{equation}
\end{proposition}
\begin{proof}
  We will repeatedly use the fact $(\diff/\diff\,x) \phi(x) = - x
  \phi(x)$ in the proof. By evaluating the differentiation in
  \eqref{eq:foldnorm-mlr}, it suffices to prove
  \[
    g(\mu) = \frac{- (x - \mu) \phi(x - \mu) - (x + \mu) \phi(x +
      \mu)}{\phi(x - \mu) + \phi(x + \mu)}
  \]
  is an increasing function of $\mu \ge 0$. Taking the derivative of
  $g(\mu)$, we have
  \[
    \begin{split}
      &[\phi(x - \mu) + \phi(x + \mu)]^2 \cdot \frac{\diff}{\diff \mu}
      g(\mu) \\
      =& \big\{[-(x - \mu)^2 + 1] \phi(x - \mu) + [(x + \mu)^2 - 1] \phi(x +
      \mu)\big\} \big[\phi(x - \mu) + \phi(x + \mu)\big] \\
      &- \big[- (x - \mu) \phi(x - \mu) - (x + \mu) \phi(x +
      \mu)\big] \big[ (x - \mu) \phi(x - \mu) - (x + \mu) \phi(x +
      \mu) \big] \\
      =& \phi(x - \mu)^2 - \phi(x + \mu)^2 + 4 \mu x \cdot \phi(x - \mu)
      \phi(x + \mu) > 0.
    \end{split}
  \]
\end{proof}


\subsection{\Cref{thm:power}}
\label{sec:crefthm:power}

\begin{proof}
Denote
\[
\limsup_{n \to \infty} \frac{\mathrm{E} [|\mathcal S_{\tau}|]}{\tau n} =
\limsup_{n \to \infty} \frac{1}{n} \sum_{i=1}^n F_i(\tau)/\tau = c.
\]
By applying the Chebyshev inequality and making use of the above
display, we get, for any $\epsilon > 0$
\[
\begin{aligned}
\limsup \mathrm{P}\left(\frac{\mathcal |S_{\tau}|}{\tau} > (c + \epsilon)n
\right) &= \limsup \mathrm{P}\left(\frac{\mathcal |S_{\tau}|}{\tau} - \frac{\mathrm{E}[\mathcal |S_{\tau}|]}{\tau} > (c + \epsilon) n - \frac{\mathrm{E}[\mathcal |S_{\tau}|]}{\tau}\right)\\
&\le \limsup \frac{\mathrm{E} \left[\frac{\mathcal |S_{\tau}|}{\tau} -
    \frac{\mathrm{E}[\mathcal |S_{\tau}|]}{\tau} \right]^2}{ \left[(c + \epsilon)n - \frac{\mathrm{E}[\mathcal |S_{\tau}|]}{\tau} \right]^2}\\
&\le \limsup \frac{n/4}{\left[(\epsilon +o(1))n\right]^2}\\
& = 0.
\end{aligned}
\]
This implies that
\[
p^{\mathrm{CB}} = \frac{\mathcal |S_{\tau}|}{\tau} \cdot \min_{1 \le i \le n}
p_i \le (c + o_p(1))\,n \min_{1 \le i \le n} p_i = (c +
o_p(1)) p^{\mathrm{B}}.
\]
The other side of the inequality can be proven similarly.

\end{proof}

\subsection{\Cref{thm:equi-corr}}
\label{sec:crefthm:equi-corr}

\begin{lemma}\label{lm:large_corr}
Let $\epsilon > 0$ be an arbitrary constant. Then Theorem
\ref{thm:equi-corr} holds for any correlation sequence
$\{\rho_l\}_{l=1}^{\infty}$ such that $\rho_l \ge \epsilon$ for all
$l$.
\end{lemma}

\begin{lemma}\label{lm:small_corr}
Theorem \ref{thm:equi-corr} holds for any correlation sequence $\{\rho_l\}_{l=1}^{\infty}$ such that $\rho_l \goto 0$.
\end{lemma}

Taking these two lemmas as given for the moment, a proof of Theorem \ref{thm:equi-corr} is readily given below.
\begin{proof}[Proof of Theorem \ref{thm:equi-corr}]
Let $\hat n_{\tau} = |\mathcal{S}_{\tau}|/\tau$. Suppose on the
contrary that Theorem is false. Then, we can pick a subsequence
$\{\rho_{s_1}, \rho_{s_2}, \ldots\}$ such that, restricted to this
subsequence,
\begin{equation}\label{eq:vio_fwer}
\mathrm{P}(\hat n_{\tau} \cdot p_{\min} \le \alpha) > (1 + c) \alpha
\end{equation}
for some constant $c > 0$.

Note that the sequence $\{\rho_{s_1}, \rho_{s_2}, \ldots\}$ must further contain a subsequence with each element bounded below by 0 or a subsequence with vanishing elements. In the former case, Lemma \ref{lm:large_corr} contradicts with \eqref{eq:vio_fwer}, and in the latter case, a contradiction arises between Lemma \ref{lm:small_corr} and \eqref{eq:vio_fwer}. Hence, such subsequence $\rho_{s_1}, \rho_{s_2}, \ldots$ should not exist at all, leading to the correctness of this theorem.

\end{proof}

\begin{proof}[Proof of Lemma \ref{lm:large_corr}]
Recognizing that the equi-correlations $\rho$ are positive, we start with the following representation
\[
Y_i \overset{d}{=} \sqrt{1 - \rho}X_i + \sqrt{\rho}W,
\]
where $X_1, \ldots, X_n, W$ are iid $\mathcal{N}(0, 1)$. Write $X_{\max} = \max\{X_1, \ldots, X_n\}$. Then
\begin{equation}\label{eq:pmin}
p_{\min} = \Phi(-\sqrt{1-\rho} X_{\max} - \sqrt{\rho}W).
\end{equation}
Making use the fact that $\Phi(-x) = (1 + o(1))\varphi(x)/x$ for $x \goto \infty$, from \eqref{eq:pmin} we get
\begin{equation}\label{eq:pmin2}
p_{\min} = (1 + o_p(1)) \frac1{\sqrt{1-\rho} X_{\max} + \sqrt{\rho}W} \varphi(\sqrt{1-\rho} X_{\max} + \sqrt{\rho}W),
\end{equation}
where the term $o_p(1)$ results from recognizing $\sqrt{1-\rho} X_{\max} + \sqrt{\rho}W \goto \infty$ as $n \goto \infty$ in probability. We proceed to bound $\varphi(\sqrt{1-\rho} X_{\max} + \sqrt{\rho}W)$.  Note that
\[
\begin{aligned}
\varphi(\sqrt{1-\rho} X_{\max} + \sqrt{\rho}W) &= \frac1{\sqrt{2\pi}} \exp\left[ -(\sqrt{1-\rho} X_{\max} + \sqrt{\rho}W
)^2/2\right]\\
&= \frac1{\sqrt{2\pi}} \e^{-I_1 - I_2 - I_3},\\
\end{aligned}
\]
where $I_1=(1-\rho) X^2_{\max}/2, I_2 = \rho W^2/2, I_3 = \sqrt{\rho(1-\rho)} X_{\max} W$. Using $X_{\max} = (1 + o_{p}(1))\sqrt{2\log n}$, we see the first term $I_1$ obeys
\[
I_1 = (1-\rho) X^2_{\max}/2 = (1-\rho) (1 + o_p(1)) \left(\sqrt{2\log n}\right)^2/2 \le (1 - \epsilon + o_p(1)) \log n.
\]
The second term satisfies $I_2 = \rho W^2 /2 = O_p(1) = o_p(I_1)$, and the last terms obeys
\[
I_3 = \sqrt{\rho(1-\rho)} X_{\max} W = O_p(\sqrt{2\log n}) = o_p(I_1).
\]
Taking these results together yields $I_1 + I_2 + I_3 = (1 + o_p(1)) I_1 \le (1 - \epsilon + o_p(1)) \log n$. Hence, we obtain
\[
\varphi(\sqrt{1-\rho} X_{\max} + \sqrt{\rho}W) \ge \frac1{\sqrt{2\pi}} \e^{-(1 - \epsilon + o_p(1)) \log n} = \frac1{\sqrt{2\pi} n^{1 - \epsilon + o_p(1)}}.
\]
Plugging the inequality above into the right-hand side of \eqref{eq:pmin2} gives
\begin{equation}\label{eq:pmin_final}
\begin{aligned}
p_{\min} &= (1 + o_p(1)) \frac{\varphi(\sqrt{1-\rho} X_{\max} + \sqrt{\rho}W)}{\sqrt{1-\rho} X_{\max} + \sqrt{\rho}W} \\
& \ge (1+o_p(1)) \frac1{\sqrt{2\pi} n^{1 - \epsilon + o_p(1)}\left[ \sqrt{1-\rho} X_{\max} + \sqrt{\rho}W \right]}\\
& = (1+o_p(1)) \frac1{2n^{1 - \epsilon + o_p(1)}\sqrt{\pi(1-\rho)\log n} }.
\end{aligned}
\end{equation}

Next, we move on to consider $\hat n_c = n_c/c$. Each $p$-value $p_i  = \Phi(-\sqrt{1-\rho} X_i - \sqrt{\rho}W)$ is below the cutoff $c$ if and only if
\[
X_i \ge - \frac{\Phi^{-1}(c) + \sqrt{\rho} W}{\sqrt{1-\rho}},
\]
which asserts
\begin{equation}\label{eq:conver}
\frac{n_c}{n} = (1 + o_p(1))\Phi\left( \frac{\Phi^{-1}(c) + \sqrt{\rho} W}{\sqrt{1-\rho}}\right).
\end{equation}

Combing \eqref{eq:pmin_final} and \eqref{eq:conver} yields
\[
\begin{aligned}
\hat n_c \cdot p_{\min} &=  (1+o_p(1))\frac{n\Phi\left( \frac{\Phi^{-1}(c) + \sqrt{\rho} W}{\sqrt{1-\rho}}\right) }{c} p_{\min}\\
&\ge  (1+o_p(1))\frac{n\Phi\left( \frac{\Phi^{-1}(c) + \sqrt{\rho} W}{\sqrt{1-\rho}}\right) }{c} \cdot \frac1{2n^{1 - \epsilon + o_p(1)}\sqrt{\pi(1-\rho)\log n} }\\
&=  (1+o_p(1))\frac{\Phi\left( \frac{\Phi^{-1}(c) + \sqrt{\rho} W}{\sqrt{1-\rho}}\right) }{2c\sqrt{\pi(1-\rho)} } \cdot \frac{n^{\epsilon + o_p(1)}}{\sqrt{\log n}}.\\
\end{aligned}
\]
Observe that the first term
\[
\frac{\Phi\left( \frac{\Phi^{-1}(c) + \sqrt{\rho}
      W}{\sqrt{1-\rho}}\right) }{2c\sqrt{\pi(1-\rho)} }
\]
is a positive random variable bounded away from 0 with high probability (though it depends on $n$), whereas the second term $n^{\epsilon + o_p(1)}/\sqrt{\log n}$ diverges to $\infty$ as $n \goto \infty$. This immediately implies
\[
\mathrm{P}(\hat n_c \cdot p_{\min} \le \alpha) \goto 0,
\]
which is stronger than what the lemma claims.
\end{proof}

\begin{proof}[Proof of Lemma \ref{lm:small_corr}]
We start by proving the fact that $\hat n_c = (1 + o_p(1))n$. First, we note that
\begin{equation}\label{eq:trivial_mean}
\mathrm{E} [\hat n_c] = n.
\end{equation}
Next, its variance is given as
\[
\begin{aligned}
\Var(\hat n_c) &= \frac{\Var(\sum_{i=1}^n \bm{1}(p_i \le c))}{c^2}\\
& = \frac{n\Var(\bm{1}(p_1 \le c)) + n(n-1) \Cov(\bm{1}(p_1 \le c), \bm{1}(p_2 \le c))}{c^2}\\
& \le \frac{n/4 + n(n-1) \Cov(\bm{1}(p_1 \le c), \bm{1}(p_2 \le c))}{c^2}.
\end{aligned}
\]
To proceed, use the fact that $\Cov(\bm{1}(p_1 \le c), \bm{1}(p_2 \le c)) = O(\rho)$. Then we get
\begin{equation}\label{eq:var}
\sqrt{\Var(\hat n_c)} =\sqrt{ n^2 O\left(\frac1{n} + \rho\right)} = o(n),
\end{equation}
which together with \eqref{eq:trivial_mean} gives
\[
\hat n_c = (1 + o_p(1))n.
\]
Hence, we get
\[
\begin{aligned}
\mathrm{P}(\hat n_c \cdot p_{\min} \le \alpha) &\le  \sum_{i=1}^n \mathrm{P}(\hat n_c \cdot p_i \le \alpha) \\
&=\sum_{i=1}^n \mathrm{P}((1 + o_p(1))n \cdot p_i \le \alpha) \\
&= \sum_{i=1}^n (1 + o(1))\frac{\alpha}{n}\\
&= (1+o(1))\alpha,
\end{aligned}
\]
as desired.

\end{proof}


\bibliographystyle{plainnat}
\bibliography{ref}


\end{document}